\let\origvec\vec
\documentclass{llncs}

\let\vec\origvec

\usepackage[utf8]{inputenc}
\usepackage[T1]{fontenc}
\usepackage[english]{babel}
\usepackage{amsfonts}         
\usepackage{amsmath}
\usepackage{amssymb}
\usepackage{stmaryrd}
\usepackage{authblk}
\usepackage{xspace}
\usepackage[misc,geometry]{ifsym}
\usepackage{graphicx}
\usepackage{ltl}
\usepackage{comment}

\newcommand{\dep}[1]{\mathord{=}\!\left(#1\right)}
\newcommand{\Tr}{\mathrm{Tr}}
\newcommand{\FO}{\mathrm{FO}}

\newcommand{\LTL}{\protect\ensuremath{\mathrm{LTL}}\xspace}

\newcommand{\QPTL}{\protect\ensuremath{\mathrm{QPTL}}\xspace}

\newcommand{\teamltl}{\protect\ensuremath{\mathrm{TeamLTL}}\xspace}

\newcommand{\hyltl}{\protect\ensuremath{\mathrm{HyperLTL}}\xspace}

\newcommand{\HQPTLP}{\protect\ensuremath{\mathrm{HyperQPTL\textsuperscript{\hskip-2pt\small +}}}\xspace}
\newcommand{\HQPTL}{\protect\ensuremath{\mathrm{HyperQPTL}}\xspace}

\newcommand{\SO}{\protect\ensuremath{\mathrm{SO}}\xspace}

\newcommand{\dfn}{\mathrel{\mathop:}=}

\newcommand{\dom}{\mathrm{dom}}

\newcommand{\LL}{\mathcal{L}}

\ifdefined\U \renewcommand{\U}{\LTLuntil} \else \newcommand{\U}{\LTLuntil} \fi
\ifdefined\W \renewcommand{\W}{\LTLweakuntil} \else \newcommand{\W}{\LTLweakuntil} \fi

\ifdefined\G \renewcommand{\G}{\LTLg} \else \newcommand{\G}{\LTLg} \fi

\newcommand{\N}{\mathbb N}

\newcommand{\traces}{\mathrm{Traces}}

\newcommand{\PSPACE}{\protect\ensuremath{\mathrm{PSPACE}}}

\renewcommand{\phi}{\varphi}
\renewcommand{\epsilon}{\varepsilon}

\newcommand{\kK}{\mathrm{K}}

\newcommand{\cneg}{\sim}

\begin{document}
\title{On the Expressive Power of  TeamLTL and First-Order Team Logic over Hyperproperties}

\author{Juha Kontinen\orcidID{0000-0003-0115-5154} \and Max Sandstr\"om\textsuperscript{(\Letter)}\orcidID{0000-0002-6365-2562}}
\institute{University of Helsinki, Department of Mathematics and Statistics, Helsinki, Finland \email{\{juha.kontinen,max.sandstrom\}@helsinki.fi}}
\maketitle

\begin{abstract}
In this article we study linear temporal logics with team semantics (\teamltl) that are novel logics for defining hyperproperties. We define Kamp-type translations of these logics into fragments of first-order team logic  and second-order logic. We also characterize the expressive power and the complexity of model-checking and satisfiability  of  team logic and second-order logic by relating them to second- and third-order arithmetic.  Our results set in a larger context the recent results of L\"uck showing that the extension of TeamLTL by the Boolean negation is highly undecidable under the so-called synchronous semantics. We also study stutter-invariant fragments of extensions of TeamLTL. 
\end{abstract}
\section{Introduction}
Linear temporal logic (\LTL) is a simple logic for formalising concepts of time. It has become important in theoretical computer science, when Amir Pnueli connected it to system verification in 1977, and within that context the logic has been studied extensively \cite{pnueli}. With regards to expressive power, a classic result by Hans Kamp from 1968 shows that \LTL is expressively equivalent to $\FO^2(<)$.

\LTL has found applications in the field of formal verification, where it is used to check whether a system fulfils its specifications. However, the logic cannot capture all of the interesting specifications a system may have, since it cannot express dependencies between its executions, known as traces. These properties, coined hyperproperties by Clarkson and Schneider in 2010, include properties important for cybersecurity such as noninterference and secure information flow \cite{clarkson}. Due to this background, extensions of \LTL have recently been the focus of research.	In order to specify hyperproperties, temporal logics like \LTL and \QPTL  have been extended with explicit trace and path quantification to define  \hyltl~\cite{DBLP:conf/post/ClarksonFKMRS14} and  \HQPTL~\cite{MarkusThesis,DBLP:conf/lics/CoenenFHH19}, respectively.

 \hyltl is one of the most extensively studied of these extensions. Its formulas are interpreted over sets of traces and the syntax extends \LTL with quantification of traces. A serious limitation of  \hyltl  is that only a fixed number of traces can be named via the quantifiers when evaluating a formula and hence  global hyperproperties  cannot be expressed in   \hyltl. 

Analogously to Kamp's theorem, \hyltl can be also related to first-order logic $\FO(<,\mathsf{E})$ by encoding a set of traces $T$ in a natural way by a first-order structure $T\times\mathbb{N}$ that, in addition to linear-orders for the traces, has  an equal level predicate $\mathsf{E}$ to synchronize time between different traces \cite{Finkbeiner017}. 

On the other hand, there are alternative approaches to extending LTL to capture hyperproperties. Team semantics is a framework in which one moves on from considering truth through single assignments to regarding teams of assignments as the linchpin for the satisfaction of a formula \cite{vaananen07,KrebsMV15}. Clearly, this framework, when applied to \LTL, provides an approach on the hyperproperties. Krebs et al. in 2018 introduced two semantics for \LTL under team semantics: the synchronous semantics and the asynchronous variant that differ on the interpretation of the temporal operators \cite{kmvz18}. In team semantics the temporal operators advance time on all traces of the current team and with the disjunction $\lor$, a team can be split into two parts during the evaluation of a formula, hence the nickname splitjunction.  

While \hyltl and other  hyperlogics have been studied extensively, many of the basic properties of 
\teamltl are still not understood. Already in \cite{kmvz18} it was shown that asynchronous \teamltl collapses to LTL  and  that the synchronous  version and \hyltl are incomparable in expressivity~\cite{kmvz18}. Furthermore,  the model checking problem of synchronous \teamltl without splitjunctions $\lor$ is in $\PSPACE$~\cite{kmvz18}. Recently it was shown by L\"uck that the complexity of satisfiability and model checking of synchronous \teamltl with Boolean negation $\cneg$ is equivalent to the decision problem of third-order arithmetic~\cite{LUCK2020} and hence highly undecidable. Furthermore, a more fine-grained analysis of the complexity of synchronous \teamltl has been obtained in \cite{VBHKF20} where a decidable fragment of \teamltl was identified and new undecidable extensions that allow restricted uses of the Boolean negation. The paper also showed that \teamltl and its extensions can be translatated to \HQPTLP, which is an extension of \hyltl by (non-uniform) quantification of propositions \cite{VBHKF20}.

In this article we analyse the expressivity and complexity of \teamltl via a different route compared to that of
 \cite{VBHKF20}. We define several translations between extensions of \teamltl and first-order team logic  
 interpreted over the structures $T\times\mathbb{N}$. The benefit of translating  \teamltl into first-order team logic is that, e.g., finding a translation for $\vee$, which is the main source of complexity of \teamltl, becomes trivial. Furthermore such  translations allow us to utilize the better understanding of the properties of logics in first-order team semantics in the study of \teamltl. 
 
 We consider both the synchronous and asynchronous \teamltl extended by the Boolean negation. In Section 3 we show that  the asynchronous version can be   translated into  $\FO^3(\dep{\ldots},\sim)$, whereas for the synchronous variant the so-called  equal-level predicate is also needed. In Section 4 we define a version of  stutter-equivalence suitable for asynchronous semantics and show that $X$-free $\teamltl(\sim)$ -formulas are stutter invariant.
In Section 5 we consider the full first-order team logic, equivalently second-order logic, as a language for hyperproperties. We show that any second-order definable hyperproperty can be expressed in third-order arithmetic and, over countable teams, in second-order arithmetic. We also show that \SO captures the trace-order invariant fragment of second-order arithmetic over countable teams. Finally in Section 6 we combine our translations with the results of L\"uck to show that any logic between the synchronous $\teamltl(\sim)$ and \SO inherits the complexity properties of $\teamltl(\sim)$. 

\section{Preliminaries}
In this section we briefly discuss the basics of linear temporal logic, its team semantic extension, and first-order team semantics. We begin with describing the classical semantics for \LTL.

Let $\mathrm{AP}$ be a set of atomic propositions. A \textit{trace} $t$ over  $\mathrm{AP}$ is an infinite sequence $t\in(2^\mathrm{AP})^\omega$.
	We denote a trace as $t=(t(i))_{i=0}^\infty$, and given $j\geq 0$ we denote the suffix of $t$ starting at the $j$th element $t[j,\infty):=(t(i))_{i=j}^\infty$.
	
	Now formulas of \LTL are defined by the grammar (where $p\in\mathrm{AP}$)
	\begin{equation*}
	\varphi:=p\mid\neg p\mid\varphi\wedge\varphi\mid\varphi\vee\varphi\mid X\varphi\mid F\varphi\mid G\varphi\mid\varphi U\varphi\mid\varphi R\varphi.
	\end{equation*}
The truth definition for this language is usually defined as follows.
\begin{definition}[Classical Semantics for \LTL]
Given a trace $t$, proposition $p\in\mathrm{AP}$, and \LTL formulas $\varphi$ and $\psi$, the semantics of linear temporal logic are as follows.
\begin{equation*}
\scriptscriptstyle
\begin{aligned}
		&t\models p \Leftrightarrow p\in t(0)\\
		&t\models \neg p \Leftrightarrow p\notin t(0)\\
		&t\models\varphi\wedge\psi \Leftrightarrow t\models\varphi\mbox{ and } t\models\psi\\
		&t\models\varphi\vee\psi \Leftrightarrow t\models\varphi\mbox{ or } t\models\psi\\
		&t\models X\varphi\Leftrightarrow t[1,\infty)\models\varphi\\
\end{aligned}
\qquad
\begin{aligned}
		&t\models F\varphi\Leftrightarrow\exists k\geq 0: t[k,\infty)\models\varphi\\
		&t\models G\varphi\Leftrightarrow\forall k\geq 0: t[k,\infty)\models\varphi\\
		&t\models\varphi U\psi\Leftrightarrow\exists k\geq 0: t[k,\infty)\models\psi\mbox{ and }\\
			&\qquad\forall k'<k: t[k',\infty)\models\varphi\\
		&t\models\varphi R\psi\Leftrightarrow\forall k\geq 0: t[k,\infty)\models\psi\mbox{ or }\\
			&\qquad\exists k'<k: t[k',\infty)\models\varphi
\end{aligned}
\end{equation*}
\end{definition}

Next we consider the team semantics of \LTL. As is usual for team semantics, we extended classical semantics by evaluating truth through a set of classical evaluations, which in this case means truth is defined by a set of traces.

A team of \teamltl is a set $T$ of traces. We denote $T[i,\infty):=\{t[i,\infty)\mid t\in T\}$ and, for $f\colon T\rightarrow \mathbb{N}$, $T[f,\infty):=\{t[f(t),\infty)\mid t\in T\}$. Below, for functions as above,  we write $f'<f$ if for all $t\in T$ it holds that $f'(t)<f(t)$.
\begin{definition}[Team Semantics for \LTL]
Suppose $T$ is a team, $p\in\mathrm{AP}$, and $\varphi$ and $\psi$ are \LTL formulas. Then the semantics of \LTL are defined as follows. The cases marked with * are the same in both the asynchronous and synchronous semantics.
\begin{equation*}
\begin{aligned}
		&T\models^* p \Leftrightarrow p\in t(0)\mbox{ for all } t\in T\\
		&T\models^* \neg p\Leftrightarrow p\notin t(0)\mbox{ for all } t\in T\\ 
		&T\models^*\varphi\wedge\psi\Leftrightarrow T\models\varphi\mbox{ and }T\models\psi\\
		&T\models^*\varphi\vee\psi\Leftrightarrow \exists T_1,T_2\subseteq T\colon \\
			&\quad T_1\cup T_2=T\mbox{ and }\\
			&\quad T_1\models\varphi\mbox{ and }T_2\models\psi\\
		&T\models^* X\varphi\Leftrightarrow T[1,\infty)\models\varphi\\
		&T\models^s F\varphi\Leftrightarrow\exists k\geq 0\colon \\
			&\quad T[k,\infty)\models\varphi\\
		&T\models^a F\varphi\Leftrightarrow \exists f\colon T\rightarrow \mathbb{N} \colon \\
			&\quad T[f,\infty)\models\varphi\\
		&T\models^s G\varphi\Leftrightarrow\forall k\geq 0\colon T[k,\infty)\models\varphi\\		
\end{aligned}
\qquad
\begin{aligned}
		&T\models^a G\varphi\Leftrightarrow\forall  f\colon T\rightarrow \mathbb{N}  \\
			&\quad T[f,\infty) \models\varphi\\
		&T\models^s\varphi U\psi\Leftrightarrow\exists k\geq 0\colon T[k,\infty)\models\psi\\
			&\quad \mbox{ and }\forall k'<k \colon T[k',\infty)\models\varphi\\
		&T\models^a\varphi U\psi\Leftrightarrow \exists f\colon T\rightarrow \mathbb{N} \colon \\
			&\quad  T[f,\infty) \models\psi\mbox{ and }\\
			&\quad\forall f'<f \colon T[f',\infty) \models\varphi\\
		&T\models^s\varphi R\psi\Leftrightarrow\forall k\geq 0\colon \\
			&\quad T[k,\infty)\models\psi\mbox{ or }\exists k'<k\colon \\
			&\quad T[k',\infty)\models\varphi\\
		&T\models^a\varphi R\psi\Leftrightarrow\forall  f\colon T\rightarrow \mathbb{N}  \\
			&\quad   T[f,\infty) \models\psi\mbox{ or }\\
			&\quad \exists f'<f\colon  T[f',\infty) \models\varphi\\
\end{aligned}
\end{equation*}
\end{definition}
We denote the asynchronous and the synchronous versions of \LTL by $\teamltl^a$ and $\teamltl^s$, respectively. In the following we consider extensions of these logics by the Boolean negation $\sim$ with the usual interpretation:
$$T\models \sim \varphi\Leftrightarrow T \not \models \varphi .  $$
We note that in the following it suffices to consider the temporal operators $X$ and $U$ as the rest can be easily defined using $\sim$ and the shorthand $\top = p \vee \neg p$ under both synchronous and asynchronous semantics:  $G \phi\equiv \sim F \sim \phi$, $F\phi \equiv \top U \phi$, $\phi R \psi \equiv \sim(\sim \phi U \sim \psi)$. 

\begin{definition}[$\FO$ under team semantics and $\FO(\dep{\ldots},\sim)$]
	Formulas of $\FO$ are defined by the grammar
	\begin{equation*}
	\varphi:=x=y\mid R(x_1,\ldots,x_n)\mid\neg x=y\mid\neg R(x_1,\ldots,x_k)\mid\varphi\wedge\varphi\mid\varphi\vee\varphi\mid\exists x\varphi\mid\forall x\varphi,
	\end{equation*}
	where $x$, $y$ and $x_1,\ldots,x_n$ are variables, and $R$ is a relation symbol of arity $n$. Formulas of $\FO(\dep{\dots},\sim)$ extend the grammar by  dependence atoms $\dep{x_1,\dots,x_n,y}$ and the Boolean negation $\sim\varphi$.
\end{definition}
First-order team semantics is defined using sets $S$  of assignments $s\colon X\rightarrow M$ with a common finite domain $X$ and an arbitrary set $M$ as co-domain. For an assignment $s$, $s[a/x]$ denotes the modified assignment that acts otherwise as $s$ except that it maps $x$ into $a$. For a so-called supplementation function  $F\colon S\to \mathcal{P}(M)\setminus\emptyset$, 
we define $$S[F/x]:=\{s[a/x]\mid s\in S \mbox{ and } a\in F(s)\}.$$ The duplication of a team with respect to variable $x$ is defined by  $S[M/x]:=\{s[m/x]\mid  m\in M, s\in S\}$. Supplementation and duplication are used to generalise existential and universal quantification, respectively, into team semantics.
\begin{definition}[Team Semantics for $\FO$]
Suppose $\mathcal{M}$ is a first-order model with domain $M$, and let $S$ be a team of $\mathcal{M}$. Suppose $n\geq 1$, and  $\varphi$ and $\psi$ are $\FO$ formulas. Then the team semantics of $\FO$ are defined by the following.
\begin{equation*}
\begin{aligned}
		&\mathcal{M}\models_S x=y \Leftrightarrow \forall s\in S\colon s(x)=s(y)\\
		&\mathcal{M}\models_S R(x_1,\ldots, x_n)\Leftrightarrow \forall s\in S\colon \\
			&\quad(s(x_1),\ldots,s(x_n))\in R^\mathcal{M}\\
		&\mathcal{M}\models_S \neg x=y \Leftrightarrow \forall s\in S\colon s(x)\neq s(y)\\
		&\mathcal{M}\models_S \neg R(x_1,\ldots,x_n)\Leftrightarrow\forall s\in S\colon \\
			&\quad(s(x_1),\ldots,s(x_n))\notin R^\mathcal{M}\\
\end{aligned}
\qquad
\begin{aligned}
		&\mathcal{M}\models_S\varphi\wedge\psi\Leftrightarrow \mathcal{M}\models_S\varphi\mbox{ and }\\
			&\quad\mathcal{M}\models_S\psi\\
		&\mathcal{M}\models_S\varphi\vee\psi\Leftrightarrow \exists S_1,S_2\subseteq S\\
			&\quad\mbox{such that }S_1\cup S_2=S\mbox{ and }\\
			&\quad \mathcal{M}\models_{S_1}\varphi\mbox{ and }\mathcal{M}\models_{S_2}\psi\\
		&\mathcal{M}\models_S\exists x\varphi\Leftrightarrow\exists F\colon S\to \mathcal{P}(M)\setminus\emptyset\\
			&\quad \mbox{such that }\mathcal{M}\models_{S[F/x]}\varphi\\
		&\mathcal{M}\models_S\forall x\varphi\Leftrightarrow \mathcal{M}\models_{S[M/x]}\varphi\\
\end{aligned}
\end{equation*}
First-order team logic $\FO(\dep{\ldots},\sim)$ extends $\FO$ with $\sim$ and dependence atoms:
\begin{equation*}
\begin{aligned}
		&\mathcal{M}\models_S \dep{x_1,\ldots,x_n,y} \Leftrightarrow \forall s_1,s_2\in S\colon \mbox{ if } s_1(x_i)=s_2(x_i)\mbox{ for all }i\in\{1,\ldots,n\},\\
		&\quad\mbox{ then }s_1(y)=s_2(y)
\end{aligned}
\end{equation*}
\end{definition}
For further details see for example \cite{vaananen07}. Note that for the so-called constancy atoms $\dep{y}$  the truth definition above amounts to requiring that $y$ has a constant value in the team. Dependence logic $\FO(\dep{\ldots})$ is known to be equi-expressive with existential second-order logic while 
 $\FO(\dep{\ldots},\sim)$ is  equi-expressive  full  second-order logic  ($\SO$) \cite{vaananen07,KontinenN11}. On the other hand,  the extensions of $\FO$
by mere constancy atoms or $\sim$ alone collapse to $\FO$ for sentences \cite{galliani13a,Luck18}.
We now define some properties of team logics relevant to the results presented in this paper. 

The formulas of  $\teamltl^a$ satisfy the following flatness property: $T\models \varphi$ if and only if $\forall t \in T \colon t \models \varphi$ while the same does not hold for the formulas of $\teamltl^s$ \cite{kmvz18}. This means that   $\teamltl^a$ without any extensions cannot define non-trivial hyperproperties.
Similarly,  $\FO$-formulas (by themselves) are also flat in the sense that for all $\mathcal{M}$ and $T$: $\mathcal{M}\models_T \varphi$ if and only if $\forall s \in T \colon \mathcal{M} \models_s \varphi$, where  $ \models_s$  refers to the standard Tarskian semantics of $\FO$ \cite{vaananen07}. On the other hand, e.g., the formula  $\dep{y}$  does not have the flatness property.
A logic $L$ with team semantics has the locality property if for all formulas $\varphi\in L$,  $\mathcal{M}$ and  $T$ it holds that $\mathcal{M}\models_T\varphi$ if and only if $\mathcal{M}\models_{T\upharpoonright\mathrm{Fr}(\varphi)}\varphi$, where $\mathrm{Fr(\varphi)}$ denotes the free variables of $\varphi$. All of the logics in the so-called lax team semantics, including  $\FO(\dep{\ldots},\sim)$, satisfy the locality property \cite{galliani12}.

\section{Embedding TeamLTL and Its Extensions into First-Order Team Logic}
In this section we define translations of \teamltl into first-order team logic. We begin by considering the translation under the asynchronous semantics.

\subsection{Asynchronous Semantics}
Let $T=\{t_j\mid j\in J\}$ be a set of traces. In order to simulate  $\teamltl^a$ and its extensions in first-order team semantics we encode  $T$ by a first-order structure $\mathcal{M}_T$ of vocabulary $\{\leq\}\cup\{P_i\mid p_i\in\mbox{AP}\}$ such that
\begin{align*}
	&Dom(\mathcal{M}_T)=T\times\mathbb{N}\\
	&\leq^{\mathcal{M}_T}=\{((t_i,n),(t_j,m))\mid i=j\mbox{ and } n\leq m\}\\
	&P^{\mathcal{M}_T}_i=\{(t_k,j)\mid p_i\in t_k(j)\}.
\end{align*}
The first positions of each of the traces of the temporal team $T$ are encoded as the values of variable $x$ by the set of assignments
$S^x_T=\{s_i\mid s_i(x)=(t_i,0) \mbox{ for all }t_i\in T\}$, which we will use as the first-order team in our translation. 
The first-order encoding of a set of traces goes back to  \cite{Finkbeiner017}.

We define a translation of $\teamltl^a$ formulas into the three-variable fragment $\FO^3(\dep{\ldots},\sim)$  of team logic. We use the variables $x,y$ and $z$ in the first-order side.
   In fact, since  $\FO^3(\dep{\ldots},\sim)$ is closed under the Boolean negation, we may assume $\sim$ also available in 
 $\teamltl^a$. Furthermore, it now suffices to consider the operators $X$ and $U$  as the rest can be easily defined using the Boolean negation.
 For readability we abbreviate the  following formula
 $$ x<y \wedge \forall z (z\le x \vee y \le z \vee (\neg x\le z\wedge \neg z\le x ))$$
defining  $y$ as the  successor of $x$  by $S(x,y)$.

Define the translation $ST_u$, for $u\in \{x,y,z\}$  via simultaneous induction as follows (we only list the formulas for $ST_x$):
\begin{equation*}
\begin{aligned}
&ST_x(p_i) = P_i(x)\\
&ST_x(\neg p_i) = \neg P_i(x)\\
&ST_x(\varphi\wedge\psi )= ST_x(\varphi)\wedge ST_x(\psi)\\
&ST_x(\varphi\vee\psi )= ST_x(\varphi)\vee ST_x(\psi)\\
\end{aligned}
\quad
\begin{aligned}
&ST_x(\sim \varphi) = \sim ST_x(\varphi)\\
&ST_x(X\varphi) = \exists y(S(x,y)\wedge ST_y(\varphi))\\
&ST_x(\varphi U\psi)=\exists y(x\leq y\wedge \dep{x,y}\wedge ST_y(\psi)\wedge \\
	&\mbox{  }\sim\exists z(x\leq z\wedge z\leq y\wedge\dep{x,z} \wedge\sim ST_z(\varphi)))
\end{aligned}
\end{equation*}
The next theorem can now be proved using induction on the formula $\varphi$.
\begin{theorem}\label{teamLTL^atoSO}
Let  $\varphi$  be a  $\teamltl^a(\sim)$- formula. Then for all non-empty $T$:
	 $$T\models\varphi\Leftrightarrow\mathcal{M}_T\models_{S^x_T} ST_x(\varphi).$$
\end{theorem}
\begin{proof}
See the appendix.
\end{proof}
As a corollary (see Appendix for the proof) we obtain that any $\teamltl^a(\sim)$-definable trace property can be defined by a sentence of  the logic  $\FO^3(\dep{\ldots},\sim)$.
\begin{corollary}\label{cor}
Let  $\varphi$  be a  $\teamltl^a(\sim)$- formula. Then there exists a sentence $\psi$ of $\FO^3(\dep{\ldots},\sim)$ such that for all non-empty $T$:
 $$T\models\varphi\Leftrightarrow\mathcal{M}_T\models \psi.$$
 \end{corollary}

\subsection*{Synchronous Semantics}
Under the synchronous team semantics  $\teamltl$ does not have the flatness property \cite{kmvz18} and it is incomparable to HyperLTL. Armed with the  translation from the previous section, we need to modify it to capture the synchronicity of the semantics on the first-order side. For this end we assume that the structure  $\mathcal{M}_T$ is equipped with the  equal-level predicate $E$ that is interpreted as follows:
$$ E^{\mathcal{M}_T}=\{((t,n),(t',m)) \mid n=m \}.  $$

Next we define a translation from $\teamltl^s(\sim)$ into $\FO^3(\dep{\dots},\sim)$ as follows:
The translation is analogous to the previous translation for the atomic propositions, $\wedge$, $\vee$, and $X$. For $U$ we define the translation as follows:
\begin{eqnarray*}
	ST^*_{x}(\varphi U\psi)&=&\exists z\exists y(\dep{z} \wedge x\le y \wedge E(z,y)\wedge  ST^*_{y}(\psi))\wedge\\
&&\sim\exists z\big(\exists x(\dep{x}\wedge  E(z,x)) \wedge  x\le z < y \wedge \sim ST^*_{z}(\varphi)\big).
\end{eqnarray*}

The proof of the following theorem is now analogous to that of Theorem \ref{teamLTL^atoSO} and Corollary \ref{cor}. 
\begin{theorem}\label{teamLTL^stoTL}
Let  $\varphi$  be a  $\teamltl^s(\sim)$- formula. Then $T\models\varphi\Leftrightarrow\mathcal{M}_T\models_{S^x_T} ST^*_x(\varphi)$ for all non-empty $T$ and there exists a $\mathsf{FO^3(\dep{\ldots},\sim)}$-sentence $\psi$ such that for all non-empty $T$
 $$T\models\varphi\Leftrightarrow\mathcal{M}_T\models \psi.$$
\end{theorem}

It is interesting to note that the role of the dependence atom in the translation of U in the asynchronous semantics can be replaced by the equal level predicate and the constancy atom in the synchronous case.

\section{Asynchronicity and Stutter Equivalence}

In  \cite{LUCK2020} L\"uck  generalised the notions of stutter-equivalence and stutter-invariance to $\teamltl^s$ \cite{LUCK2020}. In this section we define an analogous concept for the asynchronous team semantics, and show how these two conceptualizations relate to each other. 

We begin by defining the classical stuttering function.

\begin{definition}[Stuttering function]
A \textit{stuttering function} of a trace $t$ is a strictly increasing function $f\colon \mathbb{N} \to \mathbb{N}$, such that $f(0)=0$ and $t(f(k)) = \ldots = t(f(k+1)-1)$ for all $k \in \mathbb{N}$. 
\end{definition}
For functions $f\colon \mathbb{N} \to \mathbb{N}$ we denote the trace $t(f(0))t(f(1))t(f(2))\ldots$ by $t[f]$ and similarly for teams $T[f] = \{t[f] \mid t\in T\}$. Lück defined the synchronous stuttering function of a team $T$ to be a function that is a stuttering function for all traces $t \in T$ simultaneously.

The asynchronous variants of these definitions are similar, with the generalization that the stuttering functions are independent for each trace. We restrict our attention to finite  teams $T$ below.

\begin{definition}[Asynchronous stuttering function]
An asynchronous stuttering function of a team $T=\{t_1,\ldots, t_k\}$ is a function $F\colon \mathbb{N} \to \mathbb{N}^k$, such that $F(n) = (f_{t_1}(n), \ldots, f_{t_k}(n))$ for stuttering functions $f_{t_1},\ldots f_{t_k}$ of traces $t_1, \ldots, t_k \in T$.
\end{definition}
For stuttering functions $F\colon \N \to \N^k$ we denote by $t[F]$ the trace 
\begin{equation*}
t(f_t(0))t(f_t(1))t(f_t(2))\ldots,
\end{equation*}
 and furthermore for teams $T[F] = \{t[F] \mid t\in T\}$.

\begin{definition}[Asynchronous stutter-equivalence]
Teams $T,T'$ are asynchronously stutter-equivalent, in symbols $T\equiv^a_{st} T'$, if there are asynchronous stuttering functions $F$ of $T$ and $F'$ of $T'$ such that $T[F] = T'[F']$.
\end{definition}

Now it is clear that every synchronous stuttering function is also an asynchronous stuttering function, but the converse does not necessarily hold.

A formula $\varphi$ is stutter-invariant if $T\models \varphi$ if and only if $T'\models\varphi$, for all stutter equivalent teams $T$ and $T'$. Note that an asynchronous stuttering function $F$ of team $T$ induces a stuttering function   $F|S$ for any
subteam $S\subseteq T$. Next we prove analogous claims to the ones shown by Lück.

We call a function $i\colon T \to \N$ a configuration, and we define an component-wise order among configurations, i.e. for configurations $i$ and $j$ $i<j$ if and only if $i(t) < j(t)$ for all $t \in T$. The definition of $\leq$ is analogous. We consider the component-wise order because not all configurations of a team are attainable from each other with regards to the until-operator. For instance, let $T = \{t_1,t_2\}$ be a team. Now the configuration $(t_1(0),t_2(1))$ is not attainable from the configuration $(t_1(1),t_2(0))$, or vice versa. Hence we say that a configuration is attainable form another if for each trace of the former configuration the time-point is equal or later to the time-point of the same trace in the latter configuration.

\begin{lemma}\label{lemmavee}
Let $T$ and $T'$ be teams. Then $T \equiv_{st}^a T'$ if and only if, if $T = T_1 \cup T_2$, then there are subteams $T'_1,T'_2$ such that $T' = T'_1 \cup T'_2$ and $T_i \equiv_{st}^a T'_i$ for $i \in \{1,2\}$.
\end{lemma}
\begin{proof}
See the Appendix.
\end{proof}

\begin{lemma}\label{lemmaU}
Let $T$ and $T'$ be teams such that $T \equiv_{st}^a T'$, as witnessed by the stuttering functions $F$ and $G$, and let $j$ be a configuration of $T$. Then there is a configuration $i\leq j$ of $T$ such that for all $t \in T$ there is a $a_t \in \N$ such that $f_t(a_t) = i(t)$ and $T[i,\infty)\equiv_{st}^a T[j,\infty)$. Furthermore, there is a configuration $k$ of $T'$ such that $T[i,\infty)\equiv_{st}^a T'[k,\infty)$.
\end{lemma}
\begin{proof}
See the Appendix.
\end{proof}

\begin{theorem}\label{stutterInvariance}
Every $X$-free $\teamltl(\sim)$-formula is stutter invariant.
\end{theorem}
\begin{proof}
We consider formulas $\varphi \in \teamltl(\sim,U)$. Let $T$ and $T'$ be teams such that $T \equiv_{st}^a T'$, as witnessed by the functions $F$ and $G$ respectively. In this proof we say that a configuration $c$ is in accordance with a stuttering function $F$ to mean that for all $t$ there is an $n$ such that $c(t) = f_t(n)$. We prove the claim through induction on the structure of $\varphi$.
\begin{itemize}
\item For all propositional formulas the claim holds, since $F(0) = (0,\ldots,0) = G(0)$.
\item For the Boolean connectives $\wedge$ and $\sim$ the claim follows immediately from the induction hypothesis.
\item The case for the splitjunction $\vee$ is an immediate consequence of the induction hypothesis and Lemma \ref{lemmavee}.
\item For the case of the until operator $U$, we only show $\Rightarrow$, since the other direction is symmetric. Assume $T\models \psi U \theta$, i.e. there is a configuration $i$ of $T$ such that $T[i,\infty)\models \theta$ and $T[j,\infty)\models \psi$ for all configurations $j < i$. We aim to show that $T'\models \psi U \theta$. By Lemma \ref{lemmaU} we know that there is a configuration $c$ in accordance with the stuttering function $F$ such that $T[c,\infty)\models \theta$ and $c\leq i$, and there is a configuration $k$ of $T'$ such that $T[i,\infty) \equiv_{st}^a T'[k,\infty)$, which is in accordance with the stuttering function $G$. Thus by the induction hypothesis $T'[k,\infty)\models \theta$. We still need to show that $T'[l,\infty)\models \psi$ for all $l<k$. If $k(t)=0$ for all $t\in T'$, then we are done. If not, choose an arbitrary configuration $l<k$. Now by Lemma \ref{lemmaU} there is a configuration $c'$ in accordance with the stuttering function $G$ such that $c' \leq l$ and $T'[c',\infty) \equiv_{st}^a T'[l,\infty)$, and a configuration $p$ of $T$ such that $T'[c',\infty) \equiv_{st}^a T[p,\infty)$. Now $c'\leq l < k$, and thus $p < c \leq i$. Therefore $T[p,\infty)\models \psi$, and by the induction hypothesis, also $T'[l,\infty)\models\psi$.
\end{itemize}
\end{proof}

\begin{example}
We show that asynchronous stutter-equivalence does not imply synchronous stutter-equivalence. Consider the teams $T = \{\{\emptyset\}\{p\}^\omega,\{\emptyset\}\{\emptyset\}\{p\}^\omega\}$ and $T' = \{\{\emptyset\}\{p\}^\omega\}$. These two teams are asynchronously stutter-equivalent, as witnessed by the asynchronous stuttering functions $F(n) = (f_1,f_2)$ and $G(n) = n$, where $f_1 (n) = n$ and 
\begin{equation*}
f_2(n) = \begin{cases}
n &\mbox{if }n = 0 \\
n + 1 &\mbox{otherwise.} 
\end{cases}
\end{equation*}
On the other hand, synchronously stutter-equivalent teams necessarily have the same cardinality hence $T$ and $T'$ are not synchronously stutter-equivalent  \cite{LUCK2020}.

Now for all $X$-free $\teamltl^a(\sim)$ formulas $\varphi$, $T \models \varphi$ if and only if $T' \models \varphi$, however $T\nvDash^s Fp$ and $T'\models^s Fp$, which indicates that the $\teamltl^s$ formula $Fp$ cannot be expressed  by any $X$-free formula of $\teamltl^a(\sim)$. 
\end{example}

\section{SO Versus Arithmetic Definability}
The translations given in the previous section show that both $\teamltl^s(\sim)$ and $\teamltl^a(\sim)$ can be embedded into three-variable  fragments of first-order team logic (which is known to be equi-expressive with second-order logic).
In this section we compare second-order logic and  arithmetic as formalisms for defining hyperproperties.

We begin by showing that any SO-definable trace property can be also defined in third-order arithmetic. The analogous result relating $\teamltl^s(\sim)$ properties to third-order arithmetic was shown in  \cite{LUCK2020}. In arithmetic, a set $T$ can be represented by a third-order relation. Before going to the results,
we briefly discuss third-order arithmetic (see, e.g., \cite{Leivant94}). The standard model of arithmetic is denoted by $(\mathbb{N},+,\times,\le  0,1)$. A third-order type is a tuple $\tau = (n_1, \ldots, n_k)$, for natural numbers $k, n_1, \ldots, n_k \geq 1$. For each type $\tau$, we adopt  a countable  set of  $\tau$-variables
$\mathcal{V}_\tau := \{\mathfrak{a}, \mathfrak{b}, \mathfrak{c}, \ldots\}$, which are interpreted by third-order objects whose type is determined by $\tau$.  Syntactically, third-order logic extends the (more familiar) language  of second-order logic  by 
\begin{itemize}
\item new atomic formulas of the form  $\mathfrak{a}(A_1, \ldots, A_k) $, where for $\mathfrak{a}$ of type  $\tau = (n_1, \ldots, n_k)$, $A_i$ is a relation symbol of arity $n_i$ for $1\le i \le k$,
\item existential and universal quantification over third-order variables $\mathfrak{a}$. 
\end{itemize}
For a type   $\tau = (n_1, \ldots, n_k)$,  the third-order objects $\mathfrak{a}$ vary over elements $\mathcal{A}$ of the set 
$$ \mathcal{P}(\mathcal{P}(\mathbb{N}^{n_1})\times \cdots \times \mathcal{P}(\mathbb{N}^{n_k})) $$ 

The set of all formulas of third-order arithmetic, i.e., third-order formulas over the vocabulary
 $\{+, \times, 0, 1, =, \le \} $ is denoted by   $\Delta^3_0$ (and second-order arithmetic by $\Delta^2_0$). The subset of  $\Delta^3_0$ -sentences  true in  $(\mathbb{N},+,\times,\le  0,1)$ is denoted by a boldface letter
$\mathbf{\Delta^3_0}$  (analogously $\mathbf{\Delta^2_0}$). 

We are now ready to define the encoding of a team $T$ using suitable arithmetical relations. 
We  identify  the proposition $p_i$ with number $i$ and encode a  trace $t$ by a binary relation $S$ that $(j, k)\in S$  iff $p_k\in t(j)$. Now clearly a team $T$ can be  encoded by  a  third-order object $\mathcal{A}_T\subseteq \mathcal{P}(\mathbb{N}^2)$ of type $((2))$.

\begin{theorem}\label{thm1}
Let  $\phi\in SO$  be a  sentence. Then there exists a formula $\Tr (\phi)(\mathfrak{a})$ of $\Delta^3_0$ such that for all trace sets $T$:  
$$ \mathcal{M}_T \models \phi \iff  (\mathbb{N},+,\times,\le  0,1)\models \Tr (\phi)(\mathcal{A}_T / \mathfrak{a}), $$
where $\mathcal{M}_T$ is defined as in section 3.1.
\end{theorem}
\begin{proof}
See the appendix.
\end{proof}

Next we consider the special case of  countable teams $T$. In this case we can precisely characterize $SO$-definable hyperproperties arithmetically.  Assume that $T$ is countable, i.e., $T=\{t_i  \mid i\in \mathbb{N}\}$ or  $T=\{t_i \mid 0\le i \le n\}$ for some $n$. Now $T$ can be encoded by a single ternary relation  such that 
\begin{equation}\label{eq1}
 (i, j, k)\in A_T \iff p_k\in t_i(j).
 \end{equation}
 In order to encode also the cardinality of $T$, we let $A_T\subseteq \mathbb{N}^4$ to consist of the tuples
 $$(0,i,j,k)\in A_T \iff p_k\in t_i(j)  $$
 together with  $(1,n,n,n)$ if  $n=|T|$.  It is now straightforward to modify the translation $\Tr$ in such a way that $\Tr' (\phi)$ becomes a $\Delta^2_0$-formula as now an element of the domain  $T\times \mathbb{N}$ can be encoded by a pair $(i,j)$ (recall \eqref{eq1}) and a $k$-ary relation  $X\subseteq (T\times \mathbb{N})^k$ directly by a $2k$-ary  relation $R_X\subseteq \mathbb{N}^{2k}$. Define a translation   $\Tr'$ inductively as follows. Below $i$ is a definable constant.
\begin{equation*}
\begin{aligned}
&\Tr' (x = y) :=  x_1=y_1 \wedge x_2 = y_2 \\
 &\Tr' (x \leq y) :=  x_1=y_1 \wedge x_2 \leq y_2\\
 &\Tr' (P_i(x)) := R_T(0,x_1,x_2,i) \\
 \end{aligned}
 \quad
 \begin{aligned}
 &\Tr' (X(x_1,\ldots, x_k)) := R_X(x^1_1,x^1_2,\ldots, x^k_1,x^k_2)  \\
 &\Tr' (\exists x\varphi) := \exists x_1\exists x_2 (\theta_1 \wedge Tr' (\varphi))  \\
 &\Tr' (\exists X\varphi) := \exists R_X (\theta_2 \wedge\Tr' (\varphi))
\end{aligned}
\end{equation*}
Above the formula $\theta_1$ restricts the values of $x_1,x_2$ ($\theta_2$  analogously restricts   $R_X$) to range either over $\mathbb{N}^2$ or  $\{0,1,\ldots ,n\}\times \mathbb{N}$ depending on whether $T$ is infinite or finite (which can be detected by the existence of a tuple $(1,n,n,n)\in T$). The following theorem can be now proved analogously to Theorem \ref{thm1}.
\begin{theorem}\label{thm2}
Let  $\phi\in SO$  be a  sentence. Then there exists a formula $\Tr' (\phi)(R_T)$ of $\Delta^2_0$ such that for all countable trace sets $T$:  
$$ \mathcal{M}_T \models \phi \iff  (\mathbb{N},+,\times,\le  0,1)\models \Tr' (\phi)(A_T/R_T). $$
\end{theorem}
The above translation is quite simple due to the fact that only the representation of $T$ changes but the logic remains the same. It is worth noting though that in arithmetic we can, e.g., express a property
$$ T=\{t_i\ \mid t_i(i)=\{ p_i\}  \mbox{ and }  t_i(j)=\emptyset \mbox{ for } j\neq i \mbox{ and } i,j\in \mathbb{N}  \}   $$
that addresses infinitely many propositions $p_i$ whereas over  $\mathcal{M}_T$ only finitely many of the propositions can be mentioned in a formula via the relations $P_i$. Another difference between the representations is that in arithmetic the team $T$ is always naturally ordered  whereas  $\mathcal{M}_T $
carries no ordering for the traces. It turns out that these properties are the only obstacles for proving a converse of Theorem \ref{thm2}.  Below $I$ is either $\mathbb{N}$ or $\{0,1,\ldots, n\}$ for some $n$.
\begin{definition} 
Let  $\phi(R) \in \Delta^2_0$ be a formula. The formula $\phi(R)$ is called  \emph{trace-order invariant} if for all countable teams $T=(t_i)_{i\in I}$ and  all permutations $f\colon I\rightarrow I$: 
$$ (\mathbb{N},+,\times,\le  0,1)\models \phi(A_T/R) \leftrightarrow \phi(A_{T^f}/R)$$
where  $T^f=(t_{f(i)} ) _{i\in I} $.
\end{definition}

\begin{theorem}\label{thm3}
Let  $\phi(R) \in \Delta^2_0$ be  trace-order invariant and let  $\mathrm{AP}$ be finite. Then there exists a  $SO$-sentence $\psi$ such that for all countable  trace sets $T$ over $\mathrm{AP}$:  
$$ \mathcal{M}_T \models \psi \iff  (\mathbb{N},+,\times,\le  0,1)\models \phi(A_T/R). $$
\end{theorem}
\begin{proof} Let  $T=(t_i)_{i\in I}$ where $I=\{0,\ldots, n\}$ for some $n$  or $I=\mathbb{N}$. Let us assume first that the structure  $\mathcal{M}_T$ is also equipped with some linear ordering $\le_t$  of the elements $\{(t_i,0) \mid i\in I \}$. Now  the set $\{(t_0,i) \mid i\in \mathbb{N}\}$ (which is linearly ordered by $\le$) can be treated as an isomorphic copy of $\mathbb{N}$ over which we can interpret  the structure $(\mathbb{N},+,\times,\le  0,1)$ (the predicates $+$ and $\times$ can be defined using second-order existential quantification.) Now the following formulas can be constructed utilizing  the ``dual'' of this interpretation (see e.g., \cite{Immerman}):
\begin{itemize}
\item A formula $\theta(x,y,u,v)$ that defines (an isomorphic copy of) the relation $A_T$ from the information in the structure $\mathcal{M}_T$. For this an order preserving bijection $f$ between (an initial segment of) $\{(t_0,i) \mid i\in \mathbb{N}\}$ and the set  $\{(t_i,0) \mid i\in I \}$ and an equal-level predicate $E$
 can be first existentially quantified. The formula then  asserts that 
$$\bigvee_{p_k \in \mathrm{AP}} \big( x=0 \wedge v=k \wedge  \exists z(P_k(z)\wedge  f(y)\le z \wedge E(u,z) ) \big). $$
 Note that the finiteness of  $\mathrm{AP}$ is crucial for this to work and that  $f(y)\le z$ guarantees that the index of the trace of $z$ equals $y$. 

\item A sentence $\psi$ that expresses that 
$$ (\mathbb{N},+,\times,\le  0,1)\models \phi(A_T/R)$$
using the formula for $A_T$.
\end{itemize}
Finally we can get rid of the ordering $\le_t$  using second-order existential quantification as the set $\{(t,0) \mid t\in T\}$ is a definable subset of $T\times \mathbb{N}$. Now the sentence $\exists \le_t \psi$ will satisfy the claim of the theorem. 
\end{proof}

\section{Complexity of Model Checking and Satisfiability}

In this section we apply our results to characterize the complexity of model checking and satisfiability for first-order team logic and some of its variable fragments for hyperproperties (i.e., over structures $\mathcal{M}_T$). 

For the model checking problem it is  asked whether a  team of traces generated by a given finite Kripke structure satisfies a given formula. We consider Kripke structures of the form $\kK=(W, R, \eta , w_0)$, where $W$ is a finite set of states, $R\subseteq W^2$ a left-total transition relation, $\eta\colon W\rightarrow 2^{\mathrm{AP}}$ a labelling function, and $w_0\in W$ an initial state of $W$. A path $\sigma$ through $\kK$ is an infinite sequence $\sigma \in W^\omega$ such that $\sigma[0]= w_0$ and $(\sigma[i], \sigma[i + 1]) \in R$ for every $i \geq 0$. The trace of $\sigma$ is defined as $t(\sigma) \dfn \eta(\sigma[0])\eta(\sigma[1])\dots \in (2^\mathrm{AP})^\omega$.  A Kripke structure $\kK$ then induces  a trace set  $\traces(\kK) = \{t(\sigma) \mid \sigma \mbox{ is a path through $\kK$}\}$.
 
\begin{definition}
\begin{enumerate}
\item The model checking problem of a logic $\LL$ is the following decision problem:
Given a formula $\varphi\in\LL$  and a Kripke structure $K$ over $\mathrm{AP}$, determine whether $\mathit{Traces}(K) \models \varphi$,
\item The (countable) satisfiability problem of a logic $\LL$ is the following decision problem:
Given a formula $\varphi\in\LL$, determine whether $T\models \phi$ for some (countable)   $T\neq \emptyset$.
\end{enumerate}
\end{definition}

In \cite{LUCK2020} L\"uck gave a complete picture of the complexity properties of synchronous $\teamltl(\sim)$.  By combining  L\"uck's  results with ours, we are able to show the following. Below  $\FO^3(\dep{\ldots},\sim)$ is assumed to be equipped with the equal-level predicate $E$.
\begin{theorem}
\begin{enumerate}
\item\label{i} The model checking and satisfiability problems of first-order team logic $\FO(\dep{\ldots},\sim)$, its three-variable fragment  $\FO^3(\dep{\ldots},\sim)$,  and $\SO$ are equivalent to $\mathbf{\Delta^3_0}$ under logspace-reductions.
\item\label{ii}  The countable satisfiability problem of $\FO(\dep{\ldots},\sim)$, its three-variable fragment  $\FO^3(\dep{\ldots},\sim)$, and $\SO$ is equivalent to $\mathbf{\Delta^2_0}$ under logspace-reductions.
\item\label{iii}  The model checking and (countable) satisfiability problems of  $\teamltl^a(\sim)$ are logspace-reducible to  $\mathbf{\Delta^3_0}$ ($\mathbf{\Delta^2_0}$).
\end{enumerate}
\end{theorem}
\begin{proof} In each of the logics the results follow by utilizing the results of  \cite{LUCK2020} and the translations given in the previous sections. Let us consider the claim for satisfiability in  \ref{i}.  Note first that by Theorem \ref{teamLTL^stoTL} the  satisfiability problem of $\teamltl^s(\sim)$ can be easily reduced to that  of $\FO^3(\dep{\ldots},\sim)$ implying a  logspace-reduction from $\mathbf{\Delta^3_0}$ to the satisfiability problem of $\FO^3(\dep{\ldots},\sim)$.  On the other hand, any sentence $\phi$ of $\FO^3(\dep{\ldots},\sim)$ can be first translated to an equivalent $\SO$-sentence  $\phi^*$  \cite{KontinenN11}, and then,  using  Theorem \ref{thm1}, we see that $\phi$ is satisfiable iff 
$$(\mathbb{N},+,\times,\le  0,1)\models \exists \mathfrak{a} (\mathfrak{a}\neq \emptyset \wedge \Tr (\phi^*)). $$
For model checking we note that a logspace-reduction from $\mathbf{\Delta^3_0}$ to the model checking problem of $\FO^3(\dep{\ldots},\sim)$ can be obtained just like with satisfiability above. On the other hand, by Theorem 4.3 in \cite{LUCK2020} it is possible  to construct (in logspace) a  $\Delta^3_0$-formula $\psi_{\kK}(\mathfrak{a})$ defining the set $\traces(\kK)$ for any finite $\kK=(W, R, \eta , w_0)$. Hence now for any given sentence $\phi\in \SO$ and $\kK$ it holds that 
$$\mathcal{M}_ {\traces(\kK)}\models \phi \iff (\mathbb{N},+,\times,\le  0,1)\models \exists \mathfrak{a} (\psi_{\kK}(\mathfrak{a}) \wedge \Tr (\phi)). $$
Hence the model checking problem of $\SO$, $\FO(\dep{\ldots},\sim)$, and  $\FO^3(\dep{\ldots},\sim)$ reduces to  $\mathbf{\Delta^3_0}$. 

Finally claim \ref{iii} follows by  \ref{i} and \ref{ii}  and the fact that  $\teamltl^a(\sim)$ can be translated  into $\SO$ by Theorem \ref{teamLTL^atoSO}.

\end{proof}

It is worth noting that  the previous theorem can be extended to any logic $\LL$  effectively residing between:
$$\teamltl^s(\sim)\le \LL\le \SO.$$

This is a potent result, as the upper bound for many hyperlogics have not as of yet been studied.

\section{Conclusion and Future Work}
We have studied $\teamltl(\sim)$ under both synchronous and asynchronous semantics, showing through compositional translations embeddings into first-order team logic and second-order logic. Furthermore, using these  translations we were able to transfer the known complexity properties of $\teamltl^s(\sim)$ to various logics that reside between  $\teamltl^s(\sim)$  and second-order logic. Many questions remain such as:
\begin{itemize}
\item How does $\teamltl^a(\sim)$ relate to the recently defined asynchronous variant of     \hyltl \cite{DBLP:journals/corr/abs-2104-14025}?
\item Is it possible to find a non-trivial extension of $\teamltl^a$ that translates to $\FO(\sim)$ or to the extension of $\FO$ by constancy atoms? This would be interesting as both of the logics  collapse to $\FO$ for sentences.
\item What is the relationship of the logics $\teamltl^s(\sim)$ and  $\teamltl^a(\sim)$; by our result the synchronous $Fp$ cannot be expressed in $\teamltl^a(\sim)$ by any $X$-free formula and, on the other hand,  asynchronous $Fp$ cannot be expressed by any  $\teamltl^s$-formula even if the Boolean \emph{disjunction} is allowed in the formulas (see Proposition 8 in \cite{VBHKF20}).  
\end{itemize}

\subsubsection{Acknowledgements}
This research was supported by the Finnish Academy (grants 308712 and 322795).

\bibliographystyle{splncs04}
\bibliography{biblio}

\begin{thebibliography}{10}
\providecommand{\url}[1]{\texttt{#1}}
\providecommand{\urlprefix}{URL }
\providecommand{\doi}[1]{https://doi.org/#1}

\bibitem{DBLP:journals/corr/abs-2104-14025}
Baumeister, J., Coenen, N., Bonakdarpour, B., Finkbeiner, B., S{\'{a}}nchez,
  C.: A temporal logic for asynchronous hyperproperties. CoRR
  \textbf{abs/2104.14025} (2021)

\bibitem{DBLP:conf/post/ClarksonFKMRS14}
Clarkson, M.R., Finkbeiner, B., Koleini, M., Micinski, K.K., Rabe, M.N.,
  S{\'{a}}nchez, C.: Temporal logics for hyperproperties. In: {POST} 2014. pp.
  265--284 (2014)

\bibitem{clarkson}
Clarkson, M.R., Schneider, F.B.: Hyperproperties. Journal of Computer Security
  \textbf{18}(6),  1157--1210 (2010)

\bibitem{DBLP:conf/lics/CoenenFHH19}
Coenen, N., Finkbeiner, B., Hahn, C., Hofmann, J.: The hierarchy of
  hyperlogics. In: LICS 2019. pp. 1--13. {IEEE} (2019)

\bibitem{Finkbeiner017}
Finkbeiner, B., Zimmermann, M.: The first-order logic of hyperproperties. In:
  Vollmer, H., Vall{\'{e}}e, B. (eds.) STACS 2017. LIPIcs, vol.~66, pp.
  30:1--30:14. Schloss Dagstuhl - Leibniz-Zentrum f{\"{u}}r Informatik (2017)

\bibitem{galliani12}
Galliani, P.: Inclusion and exclusion dependencies in team semantics: On some
  logics of imperfect information. Annals of Pure and Applied Logic
  \textbf{163}(1),  68 -- 84 (2012)

\bibitem{galliani13a}
Galliani, P.: Epistemic operators in dependence logic. Studia Logica
  \textbf{101}(2),  367--397 (2013). \doi{10.1007/s11225-013-9478-3}

\bibitem{Immerman}
Immerman, N.: Descriptive complexity. Graduate texts in computer science,
  Springer (1999)

\bibitem{KontinenN11}
Kontinen, J., Nurmi, V.: Team logic and second-order logic. Fundam.
  Informaticae  \textbf{106}(2-4),  259--272 (2011)

\bibitem{KrebsMV15}
Krebs, A., Meier, A., Virtema, J.: A team based variant of {CTL}. In: {TIME}
  2015. pp. 140--149 (2015). \doi{10.1109/TIME.2015.11},
  \url{http://dx.doi.org/10.1109/TIME.2015.11}

\bibitem{kmvz18}
Krebs, A., Meier, A., Virtema, J., Zimmermann, M.: {Team Semantics for the
  Specification and Verification of Hyperproperties}. In: Potapov, I.,
  Spirakis, P., Worrell, J. (eds.) MFCS 2018. vol.~117, pp. 10:1--10:16.
  Schloss Dagstuhl--Leibniz-Zentrum fuer Informatik, Dagstuhl, Germany (2018)

\bibitem{Leivant94}
Leivant, D.: Higher order logic. In: Gabbay, D.M., Hogger, C.J., Robinson,
  J.A., Siekmann, J.H. (eds.) Handbook of Logic in Artificial Intelligence and
  Logic Programming, Volume 2, pp. 229--322. Oxford University Press (1994)

\bibitem{Luck18}
L{\"{u}}ck, M.: Axiomatizations of team logics. Ann. Pure Appl. Log.
  \textbf{169}(9),  928--969 (2018). \doi{10.1016/j.apal.2018.04.010}

\bibitem{LUCK2020}
L{\"{u}}ck, M.: On the complexity of linear temporal logic with team semantics.
  Theoretical Computer Science  (2020)

\bibitem{pnueli}
Pnueli, A.: The temporal logic of programs. In: 18th Annual Symposium on
  Foundations of Computer Science. pp. 46--57. {IEEE} Computer Society (1977)

\bibitem{MarkusThesis}
Rabe, M.N.: A Temporal Logic Approach to Information-Flow Control. Ph.D.
  thesis, Saarland University (2016)

\bibitem{vaananen07}
V\"a\"an\"anen, J.: Dependence Logic. Cambridge University Press (2007)

\bibitem{VBHKF20}
Virtema, J., Hofmann, J., Finkbeiner, B., Kontinen, J., Yang, F.: Linear-time
  temporal logic with team semantics: Expressivity and complexity. CoRR
  \textbf{abs/2010.03311} (2020)

\end{thebibliography}

\section*{Appendix}
\subsubsection{Proofs}
\begin{proof}[Sketch of the proof of Theorem \ref{teamLTL^atoSO}]

Let $T'\subseteq T$ and $i\colon T'\rightarrow \mathbb{N}$. We use $T'$ and $i$ as a means to  refer to  any  team that might be relevant for the evaluation of $\teamltl^a(\sim)$ formulas when starting the evaluation with $T$. On the first-order side the corresponding team will be 
$$ S^x_{T',i} := \{ s \mid s(x)= (t,i(t)) \mbox{ and } t \in T'  \}.  $$ 
We can now show using simultaneous induction on $\phi$ that for all $T'\subseteq T$, $i\colon T'\rightarrow \mathbb{N}$, and $u\in \{x,y,z\}$
$$ \{t[i(t),\infty)\mid t\in T'\}\models\varphi  \Leftrightarrow \mathcal{M}_T\models_{S^u_{T',i}} ST_u(\varphi).$$
\begin{itemize} 
\item  Assume $\varphi=  p_i$ and $T'\subseteq T$, $i\colon T'\rightarrow \mathbb{N}$ are arbitrary.
Now  
\begin{eqnarray*}
\{t[i(t),\infty)\mid t\in T'\}\models\varphi &\Leftrightarrow& p_i \in t(i(t)) \mbox{ for all } t\in T' \\
 							&\Leftrightarrow& (t,i(t))\in P^{ \mathcal{M}_T}_i  \mbox{ for all } t\in T' \\
 &\Leftrightarrow& \mathcal{M}_T\models_{S^x_{T',i}} P_i(x)\\
 		&\Leftrightarrow&  \mathcal{M}_T\models_{S^x_{T',i}} 	ST_x(\varphi) 
\end{eqnarray*}
Note that the second equivalence holds by the definition of the structure $\mathcal{M}_T$ and the third equivalence  by first-order team semantics of atomic formulas.

\item Assume $\varphi=  X\psi$ and $T'\subseteq T$ and $i\colon T'\rightarrow \mathbb{N}$ are arbitrary. Let $i^+$ be defined by $i^+(t):=i(t)+1$ for all $t$.
Now  
\begin{eqnarray*}
\{t[i(t),\infty)\mid t\in T'\}\models\varphi &\Leftrightarrow& \{t[i^+(t),\infty)\mid t\in T'\}\models\psi \\
 							 &\Leftrightarrow&  \mathcal{M}_T\models_{S^y_{T',i^+}} ST_y(\psi)\\
 		&\Leftrightarrow&  \mathcal{M}_T\models_{S^x_{T',i}} \exists y(S(x,y)\wedge ST_y(\psi))\\ 
&\Leftrightarrow&  \mathcal{M}_T\models_{S^x_{T',i}} 	ST_x(X\varphi) 
\end{eqnarray*}
The second equivalence above holds by the induction assumption for  $ST_y(\psi)$. For the the third equivalence we use the facts that the supplementation function $F$ for $y$ is uniquely determined by the formula and $x$ is not free in $ST_y(\psi)$. Note that by locality it holds that
$$ \mathcal{M}_T\models_{S^x_{T',i} [F/y]}   ST_y(\psi) \Leftrightarrow   \mathcal{M}_T\models_{S^y_{T',i^+}}   ST_y(\psi), $$ since $S^y_{T',i^+}$ is the reduct of  $S^x_{T',i}[F/y]$ to the team with domain $\{y\}$.
\end{itemize}
The proof for the connectives is straightforward and for the temporal operator $U$ it is similar to the case of $X$.
\end{proof}

\begin{proof}[Proof of Corollary \ref{cor}]
We  show that $T\models\varphi$ if and only if $\mathcal{M}_T\models \psi$, where $\psi$ is the sentence:
$$\forall x ( \exists y (y<x) \vee (\forall y (\neg y<x )\wedge ST_x(\varphi))).  $$ 
Note that
\begin{eqnarray*}
\mathcal{M}_T\models \psi &\Leftrightarrow&  \mathcal{M}_T\models_{\{\emptyset\}[\dom(\mathcal{M}_T)/x]}  \exists y  (y<x) \vee (\forall y \neg  (y<x)  \wedge ST_x(\varphi)). \\
 &\Leftrightarrow&  \mathcal{M}_T\models_{S^x_{T} }  ST_x(\varphi).
\end{eqnarray*}
In the second line the team $\{\emptyset\}[\dom(\mathcal{M}_T)/]x]$ (i.e., $\dom(\mathcal{M}_T)$) has to be  split into two disjoint parts: the subset of elements having a predecessor and to those not having a predecessor ($=S^x_{T} $).  The first disjunct is then trivially satisfied (by flatness it behaves classically) hence we arrive at the case which is equivalent to  $T\models\varphi$ by Theorem  \ref{teamLTL^atoSO}.

\end{proof}

\begin{proof}[Proof of Lemma \ref{lemmavee}]
Suppose 2. Now $T = T \cup \emptyset$, however the empty set is only stutter equivalent to itself. Thus $T \equiv_{st}^a T'$.

Suppose 1 and suppose that $T = T_1 \cup T_2$, hence we have asynchronous stuttering functions $F$ of $T$ and $G$ of $T'$, such that $T[F] = U = T'[G]$. We consider the subteams induced by the stuttering function, i.e. $T_i[F|T_i]$. Since $T[F] = T'[G]$, there exist subteams $T'_1,T'_2$ such that $T'_i[G|T'_i] = T_i[F|T_i]$. Thus $T_i \equiv_{st}^a T'_i$.

It remains to show that the subteams $T'_1$ and $T'_2$ constitute the entirety of the team $T'$. It is clear that $T'_1\cup T'_2 \subset T'$, so it remains to show the converse. Let $t' \in T'$. Now, since $T \equiv_{st}^a T'$, there exists a $t \in T$ such that $t[F|\{t\}] = t'[G|\{t'\}]$. By our assumption, the trace $t$ belongs to either $T_1$ or $T_2$. Without loss of generality we may assume that $t \in T_1$, but then $t'[G|\{t'\}] \in T'_1[G|T'_1]$. Since the team $T'$ is a set, i.e. it does not contain duplicates, we may conclude that $t' \in T'$.
\end{proof}

\begin{proof}[Proof of Lemma \ref{lemmaU}]
  By the definition of the asynchronous stuttering function, for each coordinate of $j(t)$ there exists a constant $a_t$ such that $f_t(a_t) \le j(t)$ and $t(f_t(a_t)) = t(j(t))$. Let $i$ be the configuration defined by $i(t) = f_t(a_t)$. Now we can use the stuttering function $F$ to construct stuttering functions $F'$ and $F''$ for $T[i,\infty)$ and $T[j,\infty)$ respectively. First of we define  $F'$ via  $f'_t(n) = f_t(n + a_t)$ for all $t \in T$, which clearly is a stuttering function of $T[i,\infty)$. Next we define
    \begin{align*}
      F''(n) = \begin{cases}
        (j(t))_{t\in T} &\mbox{if } n = 0\\
        F'(n) &\mbox{otherwise.}
      \end{cases}
    \end{align*}
  Since $t(f_t(a_t)) = t(i(t)) = t(j(t))$ for all $t \in T$, it follows that $T[i,\infty)[F'] = T[j,\infty)[F'']$. Thus $T[i,\infty) \equiv_{st}^a T[j,\infty)$.

  For the second claim we use the assumption that $T \equiv_{st}^a T'$. We let the configuration $i$ be as above. Now for all $n$, $t \in T$ and $t' \in T'$ it holds that $t(f_t(n)) = t'(g_{t'}(n))$. Thus there exists some configuration $k\colon T' \to \N$ such that $t'(g_{t'}(k(t'))) = t(f_t(i(t)))$, which allows us to define the stuttering function $G'$ of $T'[k,\infty)$ as $g'_{t'}(n) = g_{t'}(n + k(t'))$ for all $t' \in T'$. Clearly now $T'[k,\infty)[G'] = T[i,\infty)[F']$, and hence $T[i,\infty) \equiv_{st}^a T'[k,\infty)$.
\end{proof}

\begin{proof}[Proof of Theorem \ref{thm1}] 
We define a inductive translation $\Tr$ from  second-order logic to third order arithmetic as follows. The key ideas in the translation are:
\begin{itemize}
\item an element of the domain $T\times \mathbb{N}$ of the structure $\mathcal{M}_T$ can be uniquely identified by specifying a trace $t$ and $i\in  \mathbb{N}$. Hence, syntactically, a first-order variable $x$ can be  encoded by a pair of variables $(R_x, z_x)$ where $R_x$ is a binary relation and $z_x$ is a first-order variable;
\item a subset  of  $T\times \mathbb{N}$ is a set of pairs $(t,i)$ and hence in the translation a unary relation $X$ is encoded by a third-order variable $\mathfrak{b}_X$ of type $((2),(1))$, where  the unary relation encodes $i$ by the singleton $\{ i\}$. 
\end{itemize}
Define now a formula translation $\Tr $ as follows.  We omit below the obvious cases of the Boolean connectives and, for clarity,  we consider only unary relations $X$ on the side of $SO$. It is straightforward to write the corresponding translations also for relations of arbitrary arities. 
\begin{equation*}
\begin{aligned}
&\Tr (x = y) :=  \forall u \forall v (R_x(u,v) \leftrightarrow \\
 &\quad R_y(u,v))\wedge z_x = z_y\\
 &\Tr (x \leq y) := \forall u \forall v (R_x(u,v) \leftrightarrow \\
 &\quad R_y(u,v))\wedge z_x \leq z_y\\
 &\Tr (P_i(x)) := \mathfrak{a}(R_x) \wedge R_x(z_x,i)\\
 \end{aligned}
 \quad
 \begin{aligned}
 &\Tr (X(x)) := \exists Y(\mathfrak{b}_X (R_{x},Y)\wedge Y=\{z_x\}) \\
 &\Tr (\exists x\varphi) := \exists R_x  \exists z_x(  \mathfrak{a}(R_x) \wedge  \Tr (\varphi))\\
 &\Tr (\exists X\varphi) := \exists \mathfrak{b}_X(\forall R\forall Y( \mathfrak{b}_X(R,Y)\rightarrow  \mathfrak{a}(R)\wedge \\
 &\quad |Y|=1)\wedge \Tr (\varphi))
 \end{aligned}
\end{equation*}
In the above formulas, $i$ denotes a (definable) constant. It is now straightforward to show using induction on $\phi$ that for all $s$ and $s^*$:
\begin{equation*} 
 \mathcal{M}_T \models_s \phi \iff  (\mathbb{N},+,\times,\le  0,1)\models_{s^*} \Tr (\phi)(\mathcal{A}_T /\mathfrak{a}),
\end{equation*} 
where the interpretations $s$ and $s^*$ relate to each other as described above.
\end{proof}

\end{document}